\documentclass[submission,copyright]{eptcs}
\providecommand{\event}{TARK 2021} 
\usepackage{underscore}           

\usepackage{amsmath}
\usepackage{amssymb}
\usepackage{stmaryrd}
\usepackage{latexsym}
\usepackage{cmll}
\usepackage{amsthm}
\usepackage{bussproofs}
\usepackage{pdfsync}
\usepackage{proof}

\usepackage{thmtools,thm-restate}
\usepackage{graphicx}
\usepackage{tikz}
\usetikzlibrary{shapes}

\usepackage[english]{babel}

\usepackage[all]{xy}
\SelectTips{cm}{}

\newtheorem{definition}{Definition}
\newtheorem{lemma}{Lemma}
\newtheorem*{lemma*}{Lemma}

\newtheorem{claim}{Claim}
\newtheorem{theorem}{Theorem}
\newtheorem{obs}{Observation}

\newcommand{\comment}[1]{}

\renewcommand{\phi}{\varphi}
\renewcommand{\models}{\vDash}

\usepackage[english]{babel}
\usepackage[utf8x]{inputenc}
\usepackage[T1]{fontenc}

\usepackage{amsmath}
\usepackage{graphicx}
\usepackage[colorinlistoftodos]{todonotes}
\usepackage{tikz,pgf}
\usetikzlibrary{automata,arrows,shapes,decorations,topaths,trees,backgrounds,shadows}
\usepgfmodule{matrix}

\title{Revisiting Epistemic Logic with Names\thanks{The research of Marta B\'ilkov\'a was supported by RVO: 67985807. The research of Zo\'e Christoff and Olivier Roy was partly supported by the DFG-GACR project ``SEGA" (RO 4548/6-1). 
This publication is also part of Zo\'e Christoff's project ``Democracy on Social Networks'' (VI.Veni.201F.032) of the research programme VENI financed by the Dutch Research Council (NWO).}
}
\author{
Marta B\'ilkov\'a \institute{Czech Acad Sci, Inst Comp Sci}\email{bilkova@cs.cas.cz}
\and Zo\'e Christoff  \institute{University of Groningen}\email{z.l.christoff@rug.nl}
\and Olivier Roy \institute{University of Bayreuth}\email{olivier.roy@uni-bayreuth.de}
}
\date{}

\def\titlerunning{Revisiting Epistemic Logic with Names}
\def\authorrunning{M. B\'ilkov\'a, Z. Christoff,  \& O. Roy}

\begin{document}
\maketitle

\begin{abstract}
This paper revisits the multi-agent epistemic logic presented in~\cite{HG93}, where agents and sets of agents are replaced by abstract, intensional ``names''. We make three contributions. First, we study its model theory, providing adequate notions of bisimulation and frame morphisms, and use them to study the logic's expressive power and definability. Second, we show that the logic has a natural neighborhood semantics, which in turn allows to show that the axiomatization in~\cite{HG93} does not rely on possibly controversial introspective properties of knowledge. Finally, we extend the logic with common and distributed knowledge operators, and provide a sound and complete axiomatization for each of these extensions. These results together put the original epistemic logic with names in a more modern context and opens the door for a logical analysis of epistemic phenomena where group membership is uncertain or variable.
\end{abstract}

In~\cite{HG93}, Grove and Halpern studied a generalized version of multi-agents epistemic logic where the usual labels for agents and sets of agents are replaced by abstract names whose extension might vary from state to state.\footnote{The idea of replacing standard labels with abstract names appeared earlier, for instance in~\cite{fagin1987belief} and~\cite{moses1988programming}. The former define a notion of belief as a "society of minds" along the lines of the operator $S_n$ defined below. The latter define the notion of "everyone in a group $n$" following the same semantic idea as for the operator written $E_n$ here, and use it to define a notion of implicit common knowledge like the one later studied in~\cite{RAK}. We briefly come back to this notion in Section~\ref{sec:common}. \cite{fagin1987belief} proposes an axiomatization of belief as society of minds notions, but not together with the $E_n$ modality, as~\cite{HG93} do. \cite{moses1988programming} do not axiomatize the notion of common knowledge they put forward.} Despite being interpreted in standard multi-agents epistemic models, the resulting language does away with the familiar $K_i$ modalities, and instead contains two families of epistemic operators: $S_n$, standing for ``someone with name $n$ knows'', and $E_n$, standing for ``everyone with name $n$ knows''. 

This generalization is conceptually important. The ``names'' that index the $S_n$ and the $E_n$ modalities can refer intensionally to both individuals and groups. Since these extensions are not fixed in a given model, the logic allows to study social-epistemic phenomena that involve uncertainty or variability in the agents' identities or group membership. These phenomena are pervasive. \cite{moses1988programming,HG93} already provide convincing examples for distributed systems. Massive coordinated actions or social movements, especially online, also provide contemporary cases~\cite{bennett2014organization,coleman2014hacker}, where for instance we refer to group labels like ``Trump supporters'' or ``trolls'' without knowing exactly who the members of these groups are or even failing to know whether we, ourselves, are members of those groups.\footnote{Some of these phenomena have been studied using tools from multi-agent epistemic logic, c.f.~\cite{roy2019shared,dunin2011teamwork}.} The study in~\cite{HG93}, however, focuses on the two modalities mentioned above, and in particular leaves aside notions of common and distributed knowledge. These notions are, however, central to theories of social conventions~\cite{lewis2008convention,bicchieri2005grammar} and collective action~\cite{sep-collective-intentionality}.  \cite{moses1988programming} study a closely related notion of common knowledge, to which we come back briefly in Section~\ref{sec:common}, but do not provide an axiomatization. Distributed knowledge for intensional or indexical group names has been studied in~\cite{naumov2018everyone}, but in a more expressive language with explicit quantification. 

Epistemic logic with names is also technically interesting. Even though the idea appears earlier~\cite{von1954essay,marcus1961modalities,hintikka1962knowledge}, Grove and Halpern's contribution, as well as~\cite{grove1995naming}, has been seminal for the development of so-called term modal logic---c.f.~\cite{liberman2020dynamic} and the references therein---which in turn have helped to understand, among others, \emph{de dicto} and \emph{de re} knowledge attributions. Since many term modal logics turn out to be undecidable, one important question in that literature has been to identify decidable fragments---c.f.~\cite{shtakser2018propositional} and the references therein for the case of epistemic logic. The basic system in~\cite{HG93} is one of them. That paper, however, does not address questions of definability, invariance, expressive power, or proof theory.~\cite{padmanabha2019propositional} makes headway in that direction for a very closely related, but still non-equivalent language.  

Epistemic logic with names thus stands at the crossroad of important conceptual questions regarding group knowledge and group agency, on the one hand, and technical questions in the landscape of extended modal languages, on the other. As we show below, this generalization of epistemic logic can also be studied from the perspective of neighborhood semantics~\cite{pacuit2017neighborhood}, bringing a third tradition to bear on the understanding of this system. Although some of these areas have had contacts with each other, many dots still need to be connected, which is  what this paper sets itself to do.

After introducing the basics of epistemic logic with names, we start by formulating adequate notions of bisimulations and frame morphisms for this logic. Based on the observation that these notions are structurally similar to the corresponding notions for neighborhood frames~\cite{hansen2003monotonic,hansen2007bisimulation}, we show that the basic system can indeed be given a natural interpretation in neighborhood semantics. This allows to import a number of results regarding definability and expressive power from non-normal modal logic, as well as to show that this basic system is actually not dependent on assuming that the agent's epistemic state is represented by partitions/equivalence relations. We finally turn to group notions, showing sound and complete axiomatizations of common and distributed knowledge with names.

\section{Epistemic Logic with Names}\label{sec:background}
Epistemic logic with names replaces the familiar individual epistemic modalities $K_i$, standing for ``agent $i$ knows...'', with modalities $E_n$ for ``everybody with name $n$ knows'' and $S_n$ for ``somebody with name $n$ knows''. 
\begin{definition}[syntax]\label{def:syntax} 
The language $\mathcal{L_N}$ is defined as follows:
$$\varphi\,\, := ~p~|~\neg \varphi~|~\varphi \wedge \varphi~|~E_n\varphi~|~S_n\varphi$$
where $p\in\Phi$ with $\Phi$ a countable set of atomic propositions; $n\in N\subseteq\mathbb{N}$ is a name.
\end{definition}
 The basic idea is that these ``names'' can refer both to individuals and to groups, or even not refer to anyone at all, and that these references are intensional. Who is named by $n$ might change from state to state. In what follows we often refer to the agent(s) named by $n$ at the particular state as the "group $n$" in that state. Beside this modification, however, in~\cite{HG93} this language is interpreted in standard epistemic, i.e. "partitional" or "S5" models.
 
\begin{definition}[frames and models]\label{def:GH-frames}
Let $N$ be a given set of names. A \emph{frame over $N$} is a tuple $\mathcal{F} = (W,A,\mathcal{R},\mu)$ where:\begin{itemize}
    \item $W,A$ are (nonempty) sets of \emph{states} and \emph{agents}, respectively.
    \item $\mu:W \times N \to \mathcal{P}(A)$ is a \emph{naming function} that assigns to each world and name the set of agents that have that name in this world.
    \item $\mathcal{R}: A \to W\times W$ assigns to each agent a reflexive binary relation on $W$ such that $wR_a w'$ only if $a\in \mu(w,n)$ for some $n$. When each $R_a$ is an equivalence relation we call $F$ an \emph{epistemic} frame. We often write $R_a(w)$ for $\{v\ |\ wR_a v\}$.
\end{itemize}
An (epistemic) model over $N$ and $\Phi$ is an (epistemic) frame over $N$ together with a valuation function $\pi$ for a given set of atomic sentences $\Phi$.
\end{definition}
This definition is slightly different from the one in~\cite{HG93}. First they use an existence function $\alpha: W \to \mathcal{P}(A)$, which tells at each state which agents exist, and require that only existing agents have a name. We instead implicitly define existence using the naming function. To exist is to have a name, or be a group member. It can easily be shown that truth in the language above is invariant under this modification. In~\cite{HG93}, each $R_a$ is furthermore only allowed to connect the current state to ones where $a$ exists, which we do not assume here.

Truth and validity for this language are defined as expected, revealing the implicit quantification over agents behind the $S_n$ and the $E_n$ modalities. The clauses for the atomic formulas and Boolean connectives are standard. 
\begin{definition}[satisfaction]\label{def:satisfaction} 
Let $M$ be a model and $w\in W$: 
\begin{align*}
 M, w\vDash E_n\phi &\ \mbox{ iff\ \ for all\ \ \ \ \  }   a\in\mu(w,n):\ \forall w'\in R_a(w), \ M,  w'\vDash\phi\\
 M, w\vDash S_n\phi &\ \mbox{ iff\ \ for some }  a\in\mu(w,n): \ \forall w'\in R_a(w), \ M,  w'\vDash\phi
\end{align*}
\end{definition}
\cite{HG93} work with epistemic frames, but the logical behavior of $E_n$ and $S_n$ is much weaker than the usual $S5$ individual knowledge operators. Neither satisfy positive nor negative introspection, $E_n$ does not satisfy $T$ and $S_n$ is not normal. The following example illustrates this:

\vspace{-1,7cm}

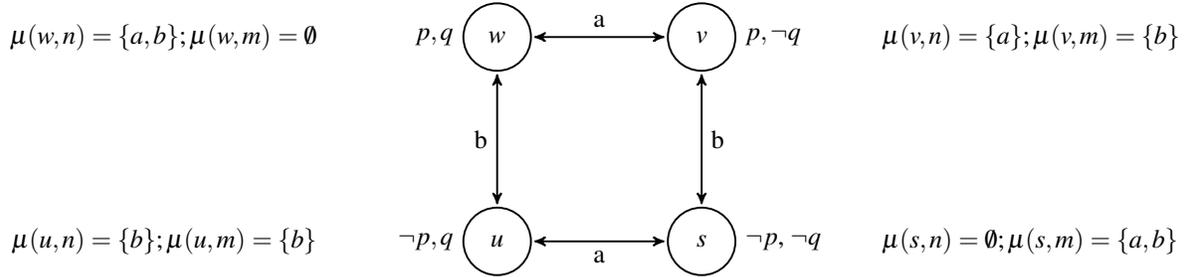
\begin{figure}[ht]
    \centering
    \resizebox{1\textwidth}{!}{
    \begin{tikzpicture}[node/.style={circle, draw=black,minimum size=10mm, thick},->,>=stealth',shorten >=1pt,shorten <=1pt,
    auto, 
    node distance=2cm,
    ]
    \tikzstyle{every state}=[draw=none,text=black]
\node[node,label=left:{$p,q$}](w){$w$};
\node[node,label=right:{$p,\neg q$}](v)[right=of w]{$v$};
\node[node,label=right:{$\neg p$, $\neg q$}](s)[below=of v]{$s$};
\node[node, label=left:{$\neg p,q$}](u)[below=of w]{$u$};
\draw[thick, <->, label] (w) -- (v) node [midway,above]{a};
\draw[thick,<->] (v) -- (s)node [midway,right]{b};
\draw[thick,<->] (s) -- (u)node [midway,below]{a};
\draw[thick,<->] (u) -- (w)node [midway,left]{b};
\node[state](muw)[left=of w]{$\mu(w,n)=\{a,b\};\mu(w,m)=\emptyset$};
\node[state](muv)[right=of v]{$\mu(v,n)=\{a\};\mu(v,m)=\{b\}$};
\node[state](mus)[right=of s]{$\mu(s,n)=\emptyset;\mu(s,m)=\{a,b\}$};
\node[state](muu)[left=of u]{$\mu(u,n)=\{b\};\mu(u,m)=\{b\}$};
    \end{tikzpicture}
    }
    \vspace{-2cm}
    \caption{An epistemic model with names. The labeled arrows represent agents $a$ and $b$’s respective indistinguishably relations (we omit reflexive arrows).}
    \label{fig:example}
\end{figure}

In the model depicted in Figure \ref{fig:example}, agent $a$ knows whether $p$, say ``Trump was impeached'', but does not know whether $q$, say ``Trump lost the election'', and vice versa for agent $b$. Consider state $w$, where both $a$ and $b$ are labeled by $n$, say ``Trump supporters'', while neither of them is labeled $m$, say ``trolls''.   Agent $a$ knows that she herself is a Trump supporter and that $b$ is either a fellow Trump supporter or a troll. In $w$, some but not all Trump supporters know that Trump was impeached, $w\models S_np \wedge \neg E_np$. And since there is no troll in $w$, no troll knows that Trump was impeached, while at the same time, trivially, all trolls know that he was impeached, in as much as they know that he was not, $w\models \neg S_mp \wedge E_m p \wedge E_m\neg p$.  In state $u$, some troll knows that Trump lost the election but it is not the case that some troll knows that some troll knows that Trump lost, $u\models S_mq\wedge\neg S_mS_mq$.  And in state $s$, no Trump supporter knows that Trump was impeached yet no supporter knows that no supporter knows it, $s\models \neg S_np\wedge \neg S_n\neg S_np$. 

The comparatively weak logical behavior of the two modalities is reflected at the axiomatic level. The system below is indeed shown in~\cite{HG93} to be sound and complete with respect to the special class of \emph{epistemic} frames, even though the axioms $4$ and $5$ are absent for both $S_n$ and $E_n$.  
\begin{definition}[axiom system $AX_\mathcal{N}$]\label{def:axiom}
The system $AX_\mathcal{N}$ comprises the following axioms and rules:
\begin{align*}
   & PL & \mbox{All instances of propositional tautologies} \\
  & MP & \mbox{From } \varphi \mbox{ and } \varphi \rightarrow \psi, \mbox{infer } \psi \\
   & T(S_n) & S_n\varphi \rightarrow \varphi \\
   &  K(E_n) & E_n\varphi \land E_n(\varphi \rightarrow \psi) \rightarrow E_n\psi \\
 & Nec(E_n) & \mbox{From } \varphi, \mbox{ infer } E_n\varphi \\
    & Int_1 & S_n\varphi \land E_n(\varphi \rightarrow \psi)\rightarrow S_n\psi \\
   & Int_2 & \lnot E_n\bot \rightarrow S_n\top
\end{align*}
\end{definition}
\cite{HG93} study further extensions of this system, to cover for instance cases where no two agents have the same name or every agent has its own proper name. They study in greater detail the case where the agents know their own names/which groups they belong to, which, interestingly, shed light on the source of introspection for standard epistemic modalities. Since in this paper we focus on the general system as defined above, we leave the discussion of these special cases for future work.

\section{Morphisms and bisimulations}\label{sec:modeltheory}

To allow for further model-theoretic considerations, we start with the definition of an adequate notion of frame morphisms and bisimulations for epistemic logic with names.
\begin{definition}\label{def:GH-morphisms}
A frame morphism from a frame $F = (W,A,\mathcal{R},\mu)$ to a frame $F'= (W',A',\mathcal{R}',\mu')$ is a map $f: W \to W'$ satisfying the following conditions:
\begin{itemize}
    \item[] (there)\ \ $\forall a\in\mu(w,n)\ \exists a'\in\mu'(f(w),n)\ R'_{a'}(f(w)) = f[R_a(w)]$
    \item[] (back)\ \ \ $\forall a'\in\mu'(f(w),n)\ \exists a\in\mu(w,n) \ R'_{a'}(f(w)) = f[R_a(w)]$
\end{itemize}
\end{definition}
In both items, $R'_{a'}(f(w)) = f[R_a(w)]$ can be equivalently split into the two usual there-and-back conditions: $w R_a v$ implies $f(w)R'_{a'}f(v)$, and $f(w)R'_{a'} w'$ implies $\exists v\in W (w R_a v \wedge f(v)=w')$. Frame validity is, as expected, preserved under frame morphisms. If each world and its image satisfy the same atomic propositions, we obtain an invariance result for the language of epistemic logic with names. $F,w \vDash \phi$ is defined as $\forall\pi\ F,w,\pi \vDash \phi$.
\begin{lemma}\label{lem:GH-morphisms-preserve-flas}
Assume that $f: F \to F'$ is a frame morphism, and valuations $\pi,\pi'$ are given so that $\pi(w) = \pi'(f(w))$ for each $w\in W$. Then for each formula $\phi \in\mathcal{L_{N}}$,
$$
F,\pi,w \vDash \phi \ \mbox{ if and only if }\ F',\pi',f(w)\vDash \phi.
$$
From this it follows that $F,w \vDash \phi \ \mbox{ implies }\ F',f(w)\vDash \phi.$
\end{lemma}
The usual model theoretic constructs can be grasped using frame morphisms: that $F'$ is a generated subframe of $F$ is defined via an injective frame morphism $f:F'\rightarrowtail F$, a morphic image frame via a surjective frame morphism $f: F\twoheadrightarrow F'$, and we can see that inclusions in a disjoint union of frames $f_i: F_i \rightarrowtail \biguplus_{i\in I}F_i$ are frame morphisms. This provides us with the usual validity preservation results. 
Regarding frame definability, we can already show that the language behaves differently from standard modal logics: 
for example, the set of formulas $\{S_n p\rightarrow p \ |\ n\in N\}$ defines the class of frames satisfying the condition
    $
    \forall x \forall n\ (\forall a\in\mu(x,n). xR_a x).
    $
The condition requires reflexive loops to exist on each state $x$ for all relations indexed by the agents named in $x$. The language \emph{cannot}, however, define reflexivity of specific individual relations. The proof is short but illustrative, so we state it explicitly.
\begin{obs}
The class of frames satisfying $\forall x. xR_a x$ for some fixed $a\in A$ is not definable in the language $\mathcal{L_{N}}$. 
\end{obs}
\begin{proof}
Consider frame $F_1$ to consist of a single state $x R_a x$ with $\mu(x,n) = \{a\}$, and frame $F_2$ to consist of a single state $x' R_b x'$ with $\mu'(x',n) = \{b\}$. Now observe that putting $f(x)=x'$ we obtain a frame morphism from a $R_a$-reflexive frame to a $R_a$-non-reflexive frame.
\end{proof}
Similar definability/undefinability examples can be constructed for many well-known properties and modal axioms, but we do not go into the general definability characterization here, as it turns out we can place the frames in the context of monotone neighborhood structures and address the phenomenon there.

When extended with valuations, frame morphisms become morphisms between models. As a simple application of invariance under morphisms between models, one can see that the modality "someone knows that" 
with truth condition 
$
w\vDash S\phi \equiv \exists a\in \alpha(w)\ \forall v\in R_a(w)\ v\vDash\phi
$ 
is not expressible in epistemic logic with names, as morphisms cannot distinguish worlds without agents from worlds without named agents\footnote{The modality $S$ is close to a modality considered in~\cite{padmanabha2019propositional}, only the one used in~\cite{padmanabha2019propositional} quantifies globally over all agents in $A$, not only those that exist in a particular state. It is not expressible in epistemic logic with names either.}.
A straightforward argument shows that extending the language $\mathcal{L_{N}}$ with this modality would allow to differentiate frames with and without an existence function mentioned above. The individual knowledge modalities $K_a$ are not expressible in $\mathcal{L_{N}}$ either, as morphisms cannot distinguish between different agents with the same name. 

Morphisms between models, via their graph, implicitly encompass the notion of bisimulation between models. This latter notion is worthwhile defining explicitly, as this will open the door to a re-interpretation of epistemic logic with names in neighborhood semantics. This definition turns out to be essentially of the same shape as in~\cite{padmanabha2019propositional}, but taking into account the naming function.
\begin{definition}\label{def:bisimulations}
A bisimulation between a model $M = (W,A,\mathcal{R},\mu,\pi)$ and a model $M'= (W',A',\mathcal{R}',\mu',\pi')$ is a binary relation $B\subseteq W\times W'$ satisfying the following conditions: $wBw '$ implies, for each $n$,:
\begin{itemize}
    \item[(0)] $\pi(w)= \pi'(w')$
    \item[(1)]  $\forall a\in \mu(w,n) \ \exists a'\in \mu'(w',n) \ (\forall u'\in R'_{a'}(w') \ \exists u\in R_a (w) \ uBu')\ \wedge\  (\forall u\in R_{a}(w) \ \exists u'\in R'_{a'} (w') \ uBu')$ 
    \item[(2)]  $\forall a'\in \mu'(w',n) \ \exists a\in \mu(w,n) \ (\forall u\in R_{a}(w) \ \exists u'\in R'_{a'} (w') \ uBu')\ \wedge\ 
     (\forall u'\in R'_{a'}(w') \ \exists u\in R_{a} (w) \ uBu')$
\end{itemize}
If $B$ is a bisimulation as above and $wBw'$, we call $(W,A,\mathcal{R},\pi,w)$ and $(W',A',\mathcal{R'},\pi',w')$ \emph{bisimilar}.
\end{definition}
As expected, bisimilarity implies modal equivalence for the language of epistemic logic with names, and the converse holds for image-finite models. These are models where for every state $w\in W$ and $n \in N$, both $\mu(w,n)$ and $R_a(w)$ are finite. 
\begin{lemma}\label{lem:bisimulationpreservesvalidity}
Assume $B$ is a bisimulation between a model $M$ and a model $M'$, and that $wBw'$. Then for each formula $\phi\in \mathcal{L_{N}}$,
$$
M,w \vDash \phi \ \mbox{ if and only if }\ M,w'\vDash \phi.
$$
Furthermore, if both $M$ and $M'$ are image-finite, then modal equivalence implies bisimilarity. 
\end{lemma}
 \begin{proof}
The first part is standard, and proceeds by induction on the complexity of the formula $\phi$. For the second part, the proof follows the argument in ~\cite{padmanabha2019propositional}.
 \end{proof}
Graphs of frame morphisms are prominent examples of bisimulations, as shown in the following Lemma. The proof is omitted for the proceedings version of the paper.
\begin{lemma}\label{lem:graphofmorphism}
Assume that $f: F \to F'$ is a frame morphism, and valuations $\pi,\pi'$ are given so that $\pi(w) = \pi'(f(w))$ for each $w\in W$. Then the graph relation $G(f) = \{(w,f(w))\ |\ w\in W\}$ is a bisimulation between a model $F = (W,A,\mathcal{R},\mu,\pi)$ and a model $F'= (W',A',\mathcal{R}',\mu',\pi')$. Moreover, functional bisimulations are graphs of frame morphisms.
\end{lemma}

The semantics of the two modalities, the definitions of frame morphism and bisimulation above, and also the condition corresponding to the $T(S_n)$ axiom, bear striking resemblance to the model theory of monotone neighborhood models, if we see the sets $\{R_a(w)\ |\ a\in\mu(w,n)\}$ as so-called core neighborhood sets~\cite{hansen2003monotonic,pacuit2017neighborhood}.
In particular, the notions of core bounded morphism and core monotone bisimulation for so-called core-complete monotone neighborhood models (cf. again~\cite{hansen2003monotonic,pacuit2017neighborhood} for the definition) are relevant here. As we will show in the following section, we can indeed view frames and models in an equivalent way as certain neighborhood structures. This should not come as a surprise, since $S_n$ are monotone non-normal modalities of the $\exists\forall$ type. Under closer scrutiny, frame morphisms are core bounded morphisms in disguise, while our definition of bisimulation is stronger than that of core monotone bisimulation, in that it requires that $B$ is full between $R_a(w)$ and $R'_{a'}(w')$ in both (1) and (2). This is because we have an additional $\forall\forall$ type of modality (namely $E_n$) in the language while the general theory of monotone neighborhood structures only considers $\exists\forall$ modalities. Interestingly, and in contrast to core monotone bisimulations, functional bisimulations in our sense correspond to graphs of frame morphisms (and thus core bounded morphisms). 

\section{Neighborhood semantics}
Beyond the similarity in their underlying notions of model-theoretic invariance, there are also two conceptual reasons to study epistemic logic with names in neighborhood semantics. First, it allows to study the modalities $E_n$ and $S_n$ as collective epistemic attitudes in their own right, attitudes that can be instantiated by a number of different assignments of agents to names or groups. Indeed, even though they are present in the semantic structures, the language of this logic makes no direct references to individuals. Different assignments of agents to a name or group $n$, i.e. different values of $\mu(w, n)$, can yield equivalent sets of statements regarding what some or all agents in $n$ know at $w$. Abstracting from the concrete frames and models defined in the previous sections and moving to neighborhood semantics allows us to study some of the variability in group membership/naming that this allows.

The second conceptual reason to study epistemic logic with names in neighborhood semantics is that it helps assessing the importance of using epistemic, i.e. partitional/S5 models in the semantics, as done in~\cite{HG93}. We already observed that in the general case neither positive nor negative introspection are valid for $E_n$ or $S_n$, even when the underlying individual relations are assumed to be transitive and Euclidean. These assumptions being controversial anyway~\cite{williamson2002knowledge,stalnaker2006logics}, this raises the question of whether the logic $AX_\mathcal{N}$ is sound and complete with respect to a larger class of frames. The completeness result provided in~\cite{padmanabha2019propositional} provides evidence that this is the case, but a precise argument for this was still missing.

\begin{definition}\label{def:nbhd}
A neighborhood frame $\mathfrak{F}$ for a given index set $I$  is a tuple $(W, \{\nu_i\}_{i \in I})$ where $W$ is a set of states and for each $i \in I$, $\nu_i: W \to \mathcal{P}\mathcal{P}(W)$ is a neighborhood function that assigns to each state $w$ a set of sets of states. We call $\nu_i(w)$ the $i$-neighborhood of $w$. Whenever for all $w$ and $X \in \nu_i(w)$, we have that $w \in X$, we say that $\nu_i$ is \emph{reflexive}. 

A neighborhood model is a neighborhood frame together with a valuation function $\pi$. When $I$ is a set $N$ of names we call $\mathfrak{F}$ neighborhood frame for $N$.
\end{definition}
Note that neighborhood frames for $N$ do not explicitly contain agents. To avoid confusion, in this section we will sometimes refer to the frames defined in Definition~\ref{def:GH-frames} as \emph{Kripke} frames, in contrast to the neighborhood frames that we have just defined. The two epistemic modalities $S_n$ and $E_n$ are interpreted using the "inexact" semantic clause, which builds in monotonicity.
\begin{definition}(satisfaction in neighborhood models)
\begin{align*}
        w\vDash S_n\phi \ &\equiv\ \exists X\in \nu_n(w)  \ X\subseteq ||\phi|| \\
        w\vDash E_n\phi \ &\equiv\ \forall X\in \nu_n(w)  \ X\subseteq ||\phi||
    \end{align*}    
\end{definition}
As a first step, we provide a simple representation of reflexive neighborhood frames into frames as defined above. Crucially these frames are not necessarily epistemic ones.
\begin{obs}\label{representation}
Let $N$ be a set of names and $\mathcal{F} = (W,A,\mathcal{R},\mu)$ be a frame for $N$. There is a reflexive neighborhood frame $\mathfrak{F}^\mathcal{F}$ for $N$ such that for all $\pi$ and $\phi\in\mathcal{L_{N}}$,
$ \mathcal{F}, \pi, w \vDash \phi \text{ iff } \mathfrak{F}^{\mathcal{F}}, \pi, w \vDash \phi  $.
Conversely, if $\mathfrak{F}  = (W, \{\nu_i\}_{n \in N})$ is a reflexive neighborhood frame for $N$ then there is a frame $\mathcal{F}^{\mathfrak{F}}$ such that 
$ \mathfrak{F}, \pi, w \vDash \phi \text{ iff } \mathcal{F}^{\mathfrak{F}}, \pi, w \vDash \phi  $.
\end{obs}
\begin{proof}
Going from Kripke frames to neighborhood frames is straightforward, putting for all $w$, $\nu_n(w) = \{R_a(w): a \in \mu(w, n)\}$. For the other direction, we put for all $w$,
$\mu(w,n) = \nu_n(w)$,
$A = \bigcup_{n \in N}\bigcup_{w \in W} \mu(w,n)$,
and 
for all $a$,
$R_a(w)=a $. The reader might want to compare with the construction in the completeness proof of~\cite{HG93}, and also with the construction of ultrafilter frames in subsection \ref{ss:algebraicduality}, where the agents that are needed to witness the truth of formulas $S_n\phi$ in at a state $w$ are constructed by identifying them with truth sets of specific formulas, namely those also containing $\phi$ together with all formulas $\psi$ for which $E_n\psi$ is true at $w$. 
%
\end{proof}
Given a reflexive neighborhood frame $\mathfrak{F}$, the Kripke frame $\mathcal{F}^{\mathfrak{F}}$ as constructed above is reflexive but neither necessarily transitive nor symmetric. Furthermore, the frame $\mathcal{F}^{\mathfrak{F}^{\mathcal{F}}}$, although modally equivalent, is not necessarily isomorphic to $\mathcal{F}$ in the sense of the first order meta-theory. However, the identity map $\iota: \mathcal{F} \to \mathcal{F}^{\mathfrak{F}^{\mathcal{F}}}$ is a frame morphism, so it is isomorphic in the sense of frame morphisms. We can in fact present the equivalence between the two kinds of semantics as an equivalence of the corresponding categories. \subsection{Categorial Equivalence between Frames and Neighborhood frames}\label{ssA:categorial-equivalence}
We consider the category $\mathcal(C)$ of frames and frame morphisms as defined in Definition \ref{def:GH-morphisms}. On the other hand we consider the category $\mathcal(C')$ of neighborhood frames of Definition \ref{def:nbhd}, and their morphisms, i.e. maps $f:W \to W'$ satisfying:
\begin{itemize}
\item[] (there-n) $X\in\nu_n(w)$ implies $f[X]\in\nu'_n(f(w))$,
\item[] (back-n) $Y\in\nu'_n(f(w))$ implies $\exists X(f[X]=Y\ \&\ X\in\nu_n(w))$.
\end{itemize}
The definition is that of core bounded morphisms from \cite[Definition 4.6]{hansen2003monotonic}. The two constructions used in the proof of Observation \ref{representation} constitute in fact functors between the two categories, we only need to say what happens to maps $f:W \to W'$: as the underlying sets $W, W'$ remain unchanged, we can literally take the same map going both ways (from $\mathcal(C)$ to $\mathcal(C')$ or back). It remains to see that (i) if $f$ is a frame morphism $f: \mathcal{F}\to \mathcal{F'}$, then it is also a (core bounded) morphism $f: \mathfrak{F}^\mathcal{F}\to \mathfrak{F}^\mathcal{F'}$, and (ii) vice versa - if $f$ is a (core bounded) morphism $f: \mathfrak{F}\to \mathfrak{F'}$, then it is also a frame morphism $f: \mathcal{F}^\mathfrak{F}\to \mathcal{F}^\mathfrak{F'}$.

For (i), assume $f$ is a frame morphism $f: \mathcal{F}\to \mathcal{F'}$. For (there-n) assume that $X\in \nu_n(w)= \{R_a(w)\ |\ a\in\mu(w,n)\}$. So, there is some $a\in\mu(w,n)$ for which $f[X] = f[R_a(w)]$, and therefore there is some $a'\in\mu'(f(w),n)$ with $f[R_a(w)]=R'_{a'}(f(w))$ by (there). This shows that $f[X] = f[R_a(w) = R'_{a'}(f(w))] \in \nu'_n(f(w))$. For (back-n) assume $Y\in\nu'_n(f(w))$, so $Y = R'_{a'}(f(w))$ for some $a'\in \mu'(f(w),n)$. Then there is $a\in\mu(w,n$ with $ f[R_a(w)]=R'_{a'}(f(w))$ by (back), so  $R_a(w)\in\nu_n(w)$ is the required set.

For (ii), assume $f$ is a (core bounded) morphism $f: \mathfrak{F}\to \mathfrak{F'}$. For (there), assume $a\in\mu(w,n) = \nu_n(w)$ is given. By (there-n), $f[a]\in \nu'_n(f(w))=\mu'(f(w),n)$, so this is our $a'$. Now $R_a(w) = a$ and $f[R_a(w)]=f(a)=a'=R'_{a'}(f(w))$ as required. For (back), assume $a'\in\mu'(f(w),n) = \nu'_n(f(w))$ is given. By (back-n) there is some $a\in \nu_n(w) = \mu(w,n)$ with $f[a]=a'$. The rest is similar as above.

It remains to observe that the maps $\iota_{\mathcal{F}}:\mathcal{F} \to \mathcal{F}^{\mathfrak{F}^\mathcal{F}}$ and $\iota_{\mathfrak{F}}:\mathfrak{F} \to \mathfrak{F}^{\mathcal{F}^\mathfrak{F}}$, given as identity on $W$, are morphisms (iso-morphisms in fact) in the respective categories. 

Observe that nothing in the above (morphisms, their properties, or the translations between the two kinds of frames) depends on reflexivity of frames. In other words, we can establish an equivalence both for frames and neighborhood frames, and for reflexive frames and neighborhood reflexive frames.

\subsection{Completeness for reflexive neighborhood frames}
With this equivalence in hand we can proceed to show that $AX_\mathcal{N}$ is indeed complete with respect to the class of reflexive neighborhood frames. Together with the representation result this gives us that this logic is actually sound and complete with respect to the class of reflexive Kripke frames augmented with a naming function.
\begin{theorem}\label{completnessnhb}
 The logic $AX_\mathcal{N}$ is sound and complete w.r.t. the class of reflexive neighborhood frames.
\end{theorem}
\begin{proof}
The proof proceeds by a standard canonical model construction for neighborhood semantics.
\end{proof}
Besides completeness for the class of reflexive Kripke frames that comes as a corollary of the completeness and representation results for neighborhood frames, this connection allows us to relate to existing results in the model and the proof theory of non-normal modal logics. We are not aware of an algebraic duality, Goldblatt-Thomason definability theorem, or van Benthem theorem existing in the literature which would apply to our case as it is\footnote{Recall that in contrast to general model theory of monotone neighborhood logics we have additional $\forall\forall$ type modalities in the language, which in particular affects the definition of standard translation or ultrafilter extensions of frames and models.}, so we present it in more details in the next section.

On the proof theoretical side, there has been to our knowledge no study of sequent calculi dedicated to epistemic logic with names specifically. \cite{lellmann2019}, however, provides a sound and complete nested sequent system that has cut elimination for Brown's "logic of ability" \cite{brown1988}. This bi-modal logic receives the same interpretation in neighborhood semantics as epistemic logic with names, with the $S_n$ and the $E_n$ modalities corresponding to Brown's "can" and "will" operators, respectively. The only difference is that Brown's neighborhood frames are not necessarily reflexive, and so the corresponding axiomatization omits $T(S_n)$. We conjecture that the system in \cite{lellmann2019} can be extended with the standard sequent rule for T, c.f.~\cite{poggiolesi2008cut}, without breaking the cut elimination result, but we leave the development of the details for future work.

\subsection{Algebraic duality}\label{ss:algebraicduality}
Any of the categories of frames described in \ref{ssA:categorial-equivalence} (we leave out reflexivity for this part) can be seen as dual to the category of the following modal algebras with names. We pick the neighborhood frames of Definition \ref{def:nbhd} to show this is so.  A modal algebra with names over $N$ is  $\mathbb{A}=(A,\wedge,\neg,\{E_n,S_n\ |\ n\in N\})$, a Boolean algebra with modalities satisfying the following equations:
\begin{align*}
E_n\top &= \top\ & \ \neg E_n\bot &= S_n\top \\
E_n(a\wedge b) &= (E_n a\wedge E_n b)\ & \ S_n a\wedge E_n b &\leq S_n(a\wedge b).
\end{align*}
It is not hard to see that this presentation is equivalent to the one obtained by simply algebraizing $AX_\mathcal{N}$ (without the $T(S_n)$ axiom). Homomorphisms of modal algebras with names are Boolean homomorphisms preserving the $S_n,E_n$ modalities. 
\par\noindent
\textbf{Ultrafilter frames}: Given a modal algebras with names $\mathbb{A}=(A,\wedge,\neg,\{E_n,S_n\ |\ n\in N\})$ we construct its ultrafilter frame over the set of ultrafilters on $\mathbb{A}$ as $\mathfrak{F}^{\mathbb{A}} = (\mathcal{U}(\mathbb{A}), \nu^A_n)$ where for $u\in \mathcal{U}(\mathbb{A})$ 
$$
\nu^A_n(u) := \{\hat{a}\cap\bigcap\limits_{E_n d\in u}\hat{d}\ |\ S_n a\in u \}.
$$
The sets $\hat{a}=\{u\in \mathcal{U}(\mathbb{A})\ |\ a\in u\}$ constitute the (clopen) basis of a topology on $\mathcal{U}(\mathbb{A})$. 

For a homomorphism $h:\mathbb{A}\to \mathbb{B}$, the inverse-image map $h^{-1}[\ ]: \mathcal{U}(\mathbb{B})\to \mathcal{U}(\mathbb{A})$ is a (bounded core) morphism from $ \mathfrak{F}^{\mathbb{B}}$ to $\mathfrak{F}^{\mathbb{A}}$. This is not immediate to see, and we will hint at the interesting part, namely that it satisfies the (there-n) condition: Assume $Y\subseteq \mathcal{U}(\mathbb{B})$ with $Y\in\nu_n^B(u)$. It means that for some $b\in B$, $Y = \{\hat{b}\cap\bigcap\limits_{E_n c\in u}\hat{c}\ |\ S_n b\in u \}$. We want to show that for some $a\in A$ with $S_n h(a)\in u$, 
$$
\{h^{-1}[v]\ |\ v\in \{\hat{b}\cap\bigcap\limits_{E_n c\in u}\hat{c}\} = \{\hat{a}\cap\bigcap\limits_{E_n h(d)\in u}\hat{d}\} \in \nu^A_n(h^{-1}[u]).
$$
By $S_n b\in u$ we know there is at least one such candidate $a$: we can consider $h^{-1}(b)$ if $b\in Rng(h)$, or $\top$ otherwise. Now, in both cases, 
$$
    v\in \{\hat{b}\cap\bigcap\limits_{E_n c\in u}\hat{c}\} \equiv
    \{b\}\cup\{c\ |\ E_n c\in u\}\subseteq v \equiv
    h^{-1}[\{b\}\cup\{c\ |\ E_n c\in u\}]\subseteq  h^{-1}[v] \equiv
    h^{-1}[v] \in \{\hat{a}\cap\bigcap\limits_{E_n h(d)\in u}\hat{d}\}.
$$
For the last equivalence, observe that $h^{-1}(b) = a$ in the first case and that $\hat{a} = \hat{\top} = \mathcal{U}(\mathbb{A})$ in the second case, and $h^{-1}[\{c\ |\ E_n c\in u\}] = \{d\ |\  E_n h(d)\in u\} $.

\par\noindent
\textbf{Complex algebras}: Given a neighborhood frame $\mathfrak{F} = (W,\nu_n) $, we construct its complex algebra as $\mathbb{A}^{\mathfrak{F}}= (\mathcal{P}W,\cap,-,\{E_n,S_n\ |\ n\in N\})$ where
\begin{align*}
E_n(X) &:= \{w\ |\ \forall Y\in \nu_n(w)\ Y\subseteq X\},\\
S_n(X) &:= \{w\ |\ \exists Y\in \nu_n(w)\ Y\subseteq X\}.
\end{align*}
For a (bounded core) morphism $f:\mathfrak{F}\to \mathfrak{G} $, the inverse-image map $f^{-1}[\ ]: \mathbb{A}^{\mathfrak{G}}\to \mathbb{A}^{\mathfrak{F}}$ is a homomorphism of modal algebras with names.

The map $\hat{.}:\mathbb{A}\to \mathbb{A}^{\mathfrak{F}^{\mathbb{A}}}$ assigning $a\mapsto \hat{a}$ is an embedding (this underlies the completeness proof for the logic if taken without the $T(S_n)$ axiom). The frame $\mathfrak{F}^{\mathbb{A}^{\mathfrak{F}}} = ue(\mathfrak{F})$ is the ultrafilter extension of $\mathfrak{F}$. With a little extra work one can prove that taking ultrafilter extension reflects frame validity: $ue(\mathfrak{F})\vDash \phi$ implies $\mathfrak{F}\vDash \phi$. Clearly, the inverse-image maps injective morphisms to surjective ones and vice versa. Also for $\mathfrak{F} = \biguplus\limits_{i\in I} \mathfrak{F}_i $, we can show that $\mathbb{A}^{\mathfrak{F}} \simeq \prod\limits_{i\in I}\mathbb{A}^{\mathfrak{F}_i} $. Putting all this to work, one can show literally by the standard argument based on the duality and Birkhoff's theorem (cf. \cite[Theorem 7.23]{hansen2003monotonic}), the following definability theorem:
\begin{theorem}
Let K be a class of neighborhood frames over $N$ which is closed under taking ultrafilter extensions. Then K is definable in the language $\mathcal{L_{N}}$ iff K is closed under disjoint unions, generated subframes and bounded morphic images, and reflects ultrafilter extensions.
\end{theorem}
%
%
\section{Common and distributed knowledge with names}\label{sec:common}
In this section we extend the language of epistemic logic with names with the two most well-known collective epistemic modalities: common and distributed knowledge. These notions are mentioned but explicitly set aside in~\cite{HG93}. \cite{moses1988programming} and after that \cite[pp.213-218]{RAK} study what they call non-rigid common and distributed knowledge with essentially the same goal as us: formulating a meaningful version of these notions in contexts where there can be uncertainty about who is in the group. As already observed in~\cite{HG93}, non-rigid common and distributed knowledge turn out to be related but subtly different from the way we define these two group attitudes. We review these differences below. Common and distributed knowledge have also been studied in term modal logic~\cite{wang2018names,naumov2018everyone}, but in languages with different expressive power than epistemic logic with names.
\subsection{Common knowledge with names}
In the context where group membership may vary from state to state, an appropriate notion of common knowledge for group $n$ must take into account not only what the agents that are members of $n$ in the current state consider possible (and what they consider others might consider possible, and so on...), but also who these group members consider might be in the group. Here we extend epistemic logic with names with a common knowledge operator in a way that is as close as possible to the standard definition of common knowledge in multi-agent epistemic logic. After presenting the semantics and returning to our running example we present a sound and complete axiomatization. We discuss briefly at the end of the section the differences between our notion and non-rigid common knowledge defined in~\cite{moses1988programming,RAK}. 

We write $C_n$ for common knowledge among members of the group \emph{named} $n$, and $\mathcal{L_{NC}}$ for the extension of language $\mathcal{L_{N}}$ with operator $C_n$. Similarly as for the standard common knowledge operators, we can unfold $C_n$ in terms of $E_n$: 
$$
C_n \phi := \bigwedge_{k\in \mathbb{N^*}} E^k_n \phi.
$$
where $E^1_n \phi := E_n\phi$ and $E^{k+1}_n \phi := E_nE^k_n\phi$.
The underlying relation $R_{n}$ is constructed as follows: $R_{n} =\{(w,v)\mid \text{ for some } i\in \mu(w,n), (w,v) \in R_i \}.$ 
We write $R_{C_n}$ for $R^+_{n}$, the transitive closure of $R_{n}$.  The semantics of the operator is given by the following expected clause: 
\begin{center}
    $M, w \models C_n \phi$\ iff\ for all $v$ such that $(w,v) \in R_{Cn} $, $M, v \models \phi$.
\end{center}
As an example, consider again the model depicted in Figure~\ref{fig:example}. 
In state $w$, it is common knowledge among Trump supporters that he was impeached or lost the election, $w\models C_n (p\lor q)$. Note, however, that this is not common knowledge in the standard sense between $a$ and $b$, even though the set of Trump supporters is exactly $\{a,b\}$ in $w$. 
Conversely, in state $v$, it is not common knowledge among trolls that Trump won the election ($v\not\models C_m\neg q$) even though the set of trolls is just $\{b\}$ and it is common knowledge in set $\{b\}$ in the standard sense. Despite these differences, semantically our new operator still corresponds to a transitive closure operation, and is therefore captured at the axiomatic level in a way similar to  the standard notion, as our axiomatization below reflects. 
\begin{definition}[System $AX_\mathcal{NC}$]
The logic $AX_\mathcal{NC}$ is the logic $AX_\mathcal{N}$ extended with the following axioms and rules for $C_n$:
\begin{align*}
   & K(C_n) & C_n(\phi\rightarrow\psi)\rightarrow (C_n\phi  \rightarrow C_n\psi) \\
   & FP & C_n\phi\rightarrow E_n(\phi \wedge C_n\phi)\\
   & Ind & \textit{From } \phi\rightarrow E_n(\phi \wedge \psi), \textit{ infer } \phi\rightarrow C_n\psi \label{Ind}\\
   & Nec(C_n) & \mbox{From } \varphi, \mbox{ infer } C_n\varphi 
\end{align*}
\end{definition}
\begin{theorem}\label{thmCncomplete}
$AX_\mathcal{NC}$ is sound and complete with respect to the class of models.
\end{theorem}
\begin{proof}
Soundness is straightforward. 
To prove completeness, we adapt the usual method of building a canonical model with maximal consistent extensions of the finite closure of a formula to circumvent non-compactness (see for instance \cite[Section 7.3]{DEL}), and combine it with the canonical model construction introduced in \cite{HG93}.
The combination is tedious, in part due to the non-factive nature of operators $C_n$ (and $E_n$) but works as expected. The heart of the proof is case 6. of Lemma \ref{lemmaCn} below. 
We start by defining an appropriate notion of closure.
\begin{definition}[Closure $cl(\chi)$]
Let $\chi\in\mathcal{L_{NC}}$ and $Sub(\phi)$ be the set of subformulas of $\phi$. 
The closure $cl(\chi)$ of $\chi$ is the smallest set such that:
\begin{enumerate}
    \item $\chi\in cl(\chi)$
    \item if $\phi\in cl(\chi)$, then $Sub(\phi)\in cl(\chi)$,
    \item if $\phi \in cl(\chi)$ and $\phi$ is not of the form $\neg\psi$, then $\neg\phi\in cl(\chi)$,
    \item $S_n( p\vee\lnot p)\in cl(\chi)$, 
    \item $E_n( p\wedge\lnot p)\in cl(\chi)$,
    \item if $E_n \phi\in cl(\chi)$, then $S_n \phi\in cl(\chi)$,
    \item if $C_n\phi\in cl(\chi)$, then $E_n\phi\in cl(\chi)$,
    \item if $C_n\phi\in cl(\chi)$, then $E_nC_n\phi\in cl(\chi)$.
\end{enumerate}
\end{definition}
\begin{lemma}
For all $\chi\in\mathcal{L_{NC}}$, $cl{(\chi)}$ is finite. 
\end{lemma}
\begin{proof}
Standard.
\end{proof}
\begin{definition}[Maximal consistent sets in $cl(\chi)$]
A set $\Gamma$ is maximal consistent in $cl(\chi)$ when:
\begin{enumerate}
    \item $\Gamma \subseteq cl(\chi)$, 
    \item $\Gamma \nvdash\bot$,
    \item there is no $\Delta$, such that $\Gamma\subset\Delta$, $\Delta\subseteq cl(\chi)$, and $\Delta \nvdash\bot$.
\end{enumerate}
\end{definition}
\begin{lemma}[Lindenbaum]
For every $\Gamma \subseteq cl(\chi)$, if $\Gamma$ is consistent, then there is a set $\Delta$, such that $\Gamma\subseteq\Delta$ and $\Delta$ is maximal consistent in $cl(\chi)$. 
\end{lemma}
\begin{proof}
The proof goes via the standard method of enumeration of all formulas in $cl(\chi)$ and sequential construction. 
\end{proof}
\begin{definition}[$n$-path,$n$-path into $\phi$]
Given a model $M= (W,A,\mathcal{R},\pi,\mu)$, a $n$-path from $w$ is a sequence $<w_0, \cdots, w_k>$ with $w_0, \cdots w_k \in W$ and $k\in\mathbb{N^*}$, such that 
$w_0=w$ and for all $0\leq m < k$, there is an agent $i\in\mu(w_m,n)$ such that $w_m R_i w_{m+1}$.
A $n$-path into $\phi$ is a $n$-path such that for all $1\leq m\leq k$, $\phi\in w_m$.
\end{definition}
We build the canonical model in a similar manner as \cite{HG93}. 
%
\begin{definition}[Canonical model $M_{cl(\chi)}$]
Given a formula $\chi\in\mathcal{L_{NC}}$, the canonical model for the closure of $\chi$ is $M_{cl(\chi)}= (W,A,\mathcal{R},\pi,\mu)$ where: 
\renewcommand{\labelitemi}{--}
\renewcommand{\labelitemii}{--}
\begin{itemize}
 \item $W$ is the set of all maximal consistent sets within $cl(\chi)$,
 \item $A$ is the set of agents, constructed as follows: 
 \begin{itemize}
     \item for every $w \in W$, and every formula $S_n\phi \in w$, define $D_{\phi,w,n}=\{\phi\}\cup\{\psi \mid E_n\psi \in w \}$, 
     \item for every such set $D_{\phi,w,n}$, the set  $a_{\phi,w,n}=\{v\in W\mid D_{\phi,w,n}\subseteq v\}$ is an agent. 
      \end{itemize}
 \item for all $w, v \in W$, and all $a\in A$, $(w,v)\in R_a$ iff $w\in a$ and $v\in a$,
  \item for all $p \in \Phi$, $\pi(p) = \{w\mid p\in w\}$, 
 
 \item $\mu$ is such that $\mu(w,n)=\{a_{\phi,w,n}\mid S_n\phi\in w\}$.\footnote{As discussed in \cite{HG93}, this construction  guarantees at the same time that all sentences of the form $S_n\phi$ have a witness (an agent named $n$ who knows $\phi$), and that this witness also knows whatever sentence everybody with name $n$ should know.} 
\end{itemize}
\end{definition}
\comment{
\begin{definition}[Canonical model $M_{cl(\chi)}$]
Given a formula $\chi\in\mathcal{L_{NC}}$, the canonical model for the closure of $\chi$ is $M_{cl(\chi)}= (W,A,\alpha,R,\pi,\mu)$ where: 
\begin{itemize}
 \item $W$ is the set of all maximal consistent sets within $cl(\chi)$,
 
 \item $A = \mathcal{P}(W)$, 
 
 \item for all $w \in W$, $A_w=\alpha(w)=\{a \mid w\in a\}$,
 
 \item for all $w, v \in W$, and all $a\in A$, $(w,v)\in R_a$ iff $w\in a$ and $v\in a$,
 
 \item for all $p \in \Phi$, $\pi(p) = \{w\mid p\in w\}$, 
 \item for every $w \in W$, and every formula $S_n\phi \in w$, define $D_{\phi,w,n}=\{\phi\}\cup\{\psi \mid E_n\psi \in w \}$. The set $D_{\phi,w,n}$ corresponds to some agent $a_{\phi,w,n}=\{v\in W\mid D_{\phi,w,n}\subseteq v\}$. Then, let $\mu$ be such that $\mu(w,n)=\{a_{\phi,w,n}\mid S_n\phi\in w\}$.\footnote{As discussed in \cite{HG93}, this guarantees at the same time that all sentences of the form $S_n\phi$ have a witness (an agent named $n$ who knows $\phi$), and that this witness also knows whatever sentence everybody with name $n$ should know.} 
 
\end{itemize}
\end{definition}
}
\begin{lemma}\label{lemmaCn}
Let $M_{cl(\chi)}= (W,A,\mathcal{R},\pi,\mu)$ be the canonical model for $cl(\chi)$. We write $\underline{w}$ for the (finite) conjunction $\bigwedge_{\phi\in w }\phi$. For all $w,v\in W$:
\begin{enumerate}
    \item if $\phi\in cl(\chi)$ and $\underline{w}\vdash\phi$, then $\phi\in w$ (i.e. $w$ is deductively closed within $cl(\chi)$)
    \item if $\lnot\phi\in cl(\chi)$, then $\lnot\phi\in w$ iff $\phi\not\in w$,
    \item if $(\phi\wedge\psi)\in cl(\chi)$, then $(\phi\wedge\psi)\in w$ iff $\phi\in w$ and $\psi\in w$,
    \item if $S_n\phi\in cl(\chi)$, then $S_n\phi\in w$ iff for all $v\in W$ such that $wR_{a_{\phi,w,n}}v$, $\phi\in v$,  
     \item if $E_n\phi\in cl(\chi)$, then $E_n\phi\in w$ iff 
     for all $\psi$ with $S_n\psi\in w$, and all $v\in W$ such that $wR_{a_{\psi,w,n}}v$, $\phi\in v$, 
   \comment{  \item if $E_n\phi\in w$, either 
 
     i) $\phi\in w$ or ii) $E_n(p\wedge\lnot p)\in w$ (and therefore $E_n\psi\in w$ for all $E_n\psi\in cl(\chi)$).
     } 
    \item if $C_n\phi\in cl(\chi)$, then $C_n\phi\in w$ iff every $n$-path from $w$ is a path into $\phi$ 
    and into $C_n\phi$. 
\end{enumerate}
\end{lemma}
\begin{proof}
Other cases are proved in the standard way; we give below the full detail of the crucial case 6. 
To prove 6., 
assume $C_n\phi\in cl(\chi)$.
\par\smallskip
For the left-right direction,
%
let $C_n\phi\in w$. 
We show that for all $n$-paths $<w, w_1, \ldots, w_k>$, and all $k\in\mathbb{N^*}$, $\phi\in w_k$ and $C_n\phi\in w_k$, by induction on $k$.

Base case ($k=1$): 
Let $v$ be such that $wR_nv$. 
By definition of $R_n$, there is some agent $a\in\mu(w,n)$ such that $wR_av$. 
By definition of $\mu(w,n)$, $a$ is $a_{\psi,w,n}$ for some $\psi$ for which $S_n\psi\in w$. 
Since $C_n\phi\in w$ and $\vdash C_n\phi \rightarrow E_n(\phi\wedge C_n\phi)$ (FP), by $MP$, $K(E_n)$ and $PL$ we obtain $\underline{w}\vdash E_n\phi$ and $\underline{w}\vdash E_nC_n\phi$. 
Since  both $E_n\phi,E_nC_n\phi \in cl(\chi)$, we know that
$E_n\phi\in w$ and $E_nC_n\phi\in w$. 

By definition of $D_{\psi,w,n}$, since $E_n\phi\in w$ and $E_nC_n\phi\in w$, then both $\phi\in D_{\psi,w,n}$ and $C_n\phi\in D_{\psi,w,n}$. 
By definition of $R_{a_{\psi,w,n}}$, since $wR_{a_{\psi,w,n}}v$, then $v\in a_{\psi,w,n}$.  
By definition of $a_{\psi,w,n}$, for all $u\in a_{\psi,w,n}$, $D_{\psi,w,n}\subseteq u$. Therefore, in particular, $D_{\psi,w,n}\subseteq v$, and since $\phi\in D_{\psi,w,n}$, and $C_n\phi\in D_{\psi,w,n}$, it follows that $\phi\in v$ and $C_n\phi\in v$ as desired.

Step case: ($k=m+1$).  
If there is a $n$-path from $w$ $<w_0,\cdots, w_ {m+1}>$, then there is a $n$-path from $w$ to $w_m$ and $w_mR_nw_{m+1}$. By the induction hypothesis, $\phi\in w_m$ and $C_n\phi\in w_m$. To show that $\phi\in w_{m+1}$ and $C_n\phi\in w_{m+1}$, it suffices to repeat the arguments used for the base case.
\par\smallskip
For the right-left direction,
%
assume that for all $n$-paths $<w,w_1,\ldots, w_k>$ and for all $k\geq 1$, $\phi\in w_k$ and $C_n\phi\in w_k$.
We show that $C_n\phi\in w$.
We will write $W_{n,\phi}$ for the set of states from which all $n$-paths are into $\phi$. 
Let $\omega:=\bigvee_{u\in W_{n,\phi}}\underline{u}$.
By our assumption, $w\in W_{n,\phi}$ which means that $\underline{w}$ is a disjunct of $\omega$. Therefore, $\vdash\underline{w}\rightarrow \omega$. 
If we show that $\vdash\omega\rightarrow E_n(\phi \wedge\omega)$, then applying the induction rule we can infer $\vdash\omega\rightarrow C_n\phi$. Since we already have $\vdash\underline{w}\rightarrow \omega$, we obtain $\vdash \underline{w}\rightarrow C_n\phi$, and therefore, since $C_n\phi\in cl({\chi})$, also
$C_n\phi\in w$.

It remains to show that $\vdash\omega\rightarrow E_n(\phi \wedge\omega)$. We will demonstrate separately that $\vdash\omega\rightarrow E_n\phi$, and that $\vdash\omega\rightarrow E_n\omega$.

To show that $\vdash\omega\rightarrow E_n\phi$,  
suppose for contradiction it is not the case. Then $\omega\wedge \lnot E_n\phi$ is consistent. This implies that there is some disjunct $\underline{u}$ of $\omega$ such that $\underline{u}\wedge\lnot E_n\phi$ is consistent. 
Since $u$ is maximal 
(and $\lnot E_n\phi\in cl({\chi})$), 
$\lnot E_n\phi\in u$. There are two cases to consider:
 
Case i): Assume $\mu(u,n)\neq\emptyset$ ($n$ is not an empty name in $u$). 
Let $a$ be an arbitrary member of $\mu(u,n)$.
By definition of $R_a$, we have $uR_au$. 
By assumption that all $n$-paths from $u$ are into $C_n\phi$, we know that $C_n\phi\in u$. 
Since $\vdash C_n\phi\rightarrow E_n\phi$ 
and $E_n\phi\in cl(\chi)$, we know that $E_n\phi\in u$. 
But then, since both $E_n\phi\in u$ and $\lnot E_n\phi\in u$, $u$ is not consistent. A contradiction. 

Case ii) Assume $\mu(u,n)=\emptyset$ ($n$ is an empty name in $u$). 
By definition of $\mu(u,n)$, this implies that there is no $\psi$ such that $S_n\psi\in u$. 
In particular, $S_n ( p\vee\lnot p)\notin u$. 
Therefore, $\lnot S_n (p\vee\lnot p)\in u$. 
By axiom $Int_2$, we obtain $\vdash \lnot S_n (p\vee\lnot p) \rightarrow E_n (p\wedge\lnot p)$, 
and $E_n(p\wedge\lnot p)\in u$. 
And since $\vdash E_n(p\wedge\lnot p)\rightarrow E_n(\phi\wedge\lnot \phi)$ and $\vdash E_n(\phi\wedge\lnot \phi) \rightarrow E_n\phi$, 
we have $E_n\phi\in u$. 
But then, again, both $E_n\phi\in u$ and $\lnot E_n\phi\in u$, and so $u$ is not consistent. A contradiction. 
This concludes the proof that $\vdash\omega\rightarrow E_n\phi$.

We now turn to the proof of $\vdash\omega\rightarrow E_n\omega$.
Suppose for contradiction it is not the case.
Then $\omega\wedge \lnot E_n\omega$ is consistent, 
and there must be some disjunct $\underline{u_0}$ of $\omega$ such that $\underline{u_0}\wedge\lnot E_n\omega$ is consistent.
We need to construct $v\in W$ (a MCS in $cl(\chi)$) such that $\forall u\in W_{n,\phi}\ \exists \psi_u\in u$ with $\psi_u\notin v$, and we need to find an agent $a\in \mu(u_0,n)$ with $u_0R_a v$. 

Observe first that $S_n\top\in u_0$: if not, then $\neg E_n\bot\notin u_0$ and therefore $E_n\bot\in u_0$, and consequently $u_0\vdash E_n\omega$ which is a contradiction. Therefore we take $a_{\top,u_0,n}$ to be the chosen agent in $\mu(u_0,n)$. Next, observe that $D_{\top,u_0,n}\subseteq u_0$: this follows by $S_n\top \wedge E_n\gamma \vdash S_n\gamma \vdash \gamma$ for each $E_n\gamma\in u_0$. This means we need $D_{\top,u_0,n}\subseteq v$ to ensure that $u_0R_a v$.
To sum up, $v$ needs to include all formulas in $D_{\top,u_0,n}$, and for each $u\in W_{n,\phi}$ it has to pick a formula $\psi_u\in u$ to exclude (i.e. $\neg\psi_u\in v$). We need to prove that at least one such choice is consistent: assume for a contradiction that it is not so. This means that for each choice $C_i = \{\psi_u\ |\ u\in W_{n,\phi}\}$ ($i$ ranges $1\ldots \Pi_{u\in W_{n,\phi}}|u|$), we have that $D_{\top,u_0,n}\vdash \bigvee C_i$. Therefore $D_{\top,u_0,n}\vdash \bigwedge_i\bigvee C_i$ and, by the distributive law, $D_{\top,u_0,n}\vdash \omega$. But this entails that $u_0\vdash E_n\omega$, contradicting our original assumption. So, one such choice is possible to make, and by the Lindenbaum lemma we can construct $v\in W$ as required.

We can deduce that $v\notin W_{n,\phi}$ (because for any $v\in W_{\phi,n},\  \vdash\underline{v}\rightarrow\omega$). 
Since $v\notin W_{n,\phi}$, there is an $n$-path from $v$ into some state $s$ such that $\phi\notin s$. 
But any $n$-path from $v$ is an $n$-path from $w$ too, and so there is an $n$-path from $w$ which is not into $\phi$. This contradicts our initial assumption, and  concludes the proof that $\vdash\omega\rightarrow E_n\omega$.

Therefore, since we have $\vdash\omega\rightarrow E_n\phi$ and $\vdash\omega\rightarrow E_n\omega$, we obtain $\vdash\omega\rightarrow E_n(\phi\wedge\omega)$ (by $PL$, $Nec(E_n)$, and $K(E_n)$). 
Finally, by inference rule $Ind$, $\vdash\omega\rightarrow C_n\phi$ and since $\vdash \underline{w}\rightarrow \omega$, we obtain $\vdash \underline{w}\rightarrow C_n\phi$, which in turns implies that $C_n\phi\in w$. 
\end{proof}
\begin{lemma}[Truth Lemma]\label{truthlemma}
Let $\psi\in\mathcal{L_{NC}}$, and $M_{cl(\psi)}=(W,A,\mathcal{R},\pi,\mu)$ canonical for $cl(\psi)$. For all $\phi\in cl(\psi)$ : 
$$M_{cl(\psi)},w\models\phi \ \textit{ if and only if }\ \phi\in w.$$
\end{lemma}
\begin{proof}
By induction on the length of $\phi$. Boolean cases are standard, and cases for $E_n\psi$ and $S_n\psi$ are treated similarly as in \cite{HG93}, we only give the detail for the additional case where $\phi:= C_n\psi$. 
\begin{center}
\begin{tabular}{lcr}
  $M,w\models C_n\psi$   & iff  & (by the semantics of $C_n$)\\
    for all $v\in R_{C_n}(w)$, $M,v\models\psi$ & iff  & (by IH)\\
    for all $v\in R_{C_n}(w)$, $\psi\in v$ & iff  & (by def of $n$-paths) \\
    all $n$-paths from $w$ are into $\psi$ & iff  & (by Lemma \ref{lemmaCn}) \\
    $C_n\psi\in w$. & & \\[-4.5ex]
\end{tabular}
\end{center}
\end{proof}
A standard contraposition argument allows to conclude the proof of Theorem \ref{thmCncomplete}.
\end{proof}
We finish this section by comparing briefly the notion of common knowledge studied here with non-rigid common knowledge as defined in~\cite{moses1988programming,RAK}. The models used to define non-rigid common knowledge are essentially the same as the models of epistemic logic with names. The definition of non-rigid common knowledge is, however, based on a different individual modality, written $B_i^n$, which is informally described as "agent $i$ knows/believes that if $i$ is in $n$, then $\phi$". The original motivation for this reading was to study cases where a process knows that if it is correct, then $\phi$ holds. Translated into our semantics this would have the following truth condition.
     $$M, w \vDash B_i^n\phi\ \text{ iff }\ \forall w'. wR_a w' \wedge a \in \mu(w',n) \rightarrow M, w' \vDash \phi. $$
An easy argument shows that $S_n\phi$ implies $B_i^n\phi$ for some agents $i$ that are members of $n$ in the current situation, but not the other way around.~\cite{HG93} show that the corresponding $S'_n$ and $E'_n$ can be axiomatized using the same axioms as the original $S_n$ and $E_n$. On the basis of this we conjecture that the corresponding notion of common knowledge could be axiomatized as the one defined here, but we leave this for future work.
\subsection{Distributed knowledge with names}\label{sec:distributed}
For distributed knowledge we propose a generalization of the $S_n$ modality. Instead of requiring the existence of one \emph{agent} with name $n$ that knows $\phi$, $D_n \phi$ is true whenever there exists a (non-empty) \emph{sub-group} of $n$ such that pooling the information of the agents in that sub-group would entail $\phi$. 
\begin{definition}[Satisfaction, $D_n$]
    $$M, w \vDash D_n \phi\ \text{ iff }\   \exists X \subseteq \mu(w,n) , X \neq \emptyset \wedge \forall v \in \bigcap_{i \in X}R_i(w), M, v \models \phi$$
\end{definition}
\noindent This version of distributed knowledge, unlike $S_n$, turns out to be closed under conjunction. The intuitive reason for this is that sub-groups can be merged. If a sub-group $g$ of $n$ distributely knows that $p$ and another sub-group $g'$ of $n$ distributely knows that $q$, when the union $g \cup g'$ of these sub-groups effectively pools these two pieces of information, leading to distributed knowledge of $p \wedge q$. The logic of $D_n$ is not completely normal, though, as like for $S_n$ it does not validate necessitation. 
\begin{definition}[System $AX_\mathcal{ND}$] The logic $AX_\mathcal{ND}$ is the logic $AX_\mathcal{N}$, extended with the following axioms:
\begin{equation*}
    \begin{array}{ccc}
   D_n \phi \wedge D_n (\phi \rightarrow \psi) \rightarrow D_n \psi  &  & K_D\\
     
 S_n \phi \rightarrow D_n \phi & & Inclusion \\
 D_n \phi \rightarrow \phi & & T_D \\
D_n \phi \wedge E_n (\phi \rightarrow \psi) \rightarrow D_n \psi & & Interaction
\end{array}
\end{equation*}
\end{definition}
\begin{theorem}\label{completeDKnowledge}
The logic $AX_\mathcal{ND}$ is sound and complete w.r.t. the class of frames over $N$ and the language containing the modality $D_n$. 
\end{theorem}
\begin{proof}
The proof proceeds by adapting the usual copy-and-splitting technique for completeness with intersection modalities (c.f.~\cite{wang2020simpler} for a recent overview). 
\end{proof}
Distributed knowledge with names is, like its standard counterpart, not invariant under the notion of bisimulation in Definition~\ref{def:bisimulations}. It can however be extended in the standard way to also cover intersections of modalities.
\section{Conclusion} \label{sec:conclusion}
Adding names to standard epistemic logic lifts one of the most fundamental idealizations, alongside logical omniscience, of classical logic for knowledge and belief: the facts that, on the one hand, the agents' names are common knowledge and, on the other, that groups are defined extensionally. This idealization is perhaps even more problematic than logical omniscience to the extent that it appears unrealistic even if we give a normative interpretation to multi-agents epistemic logic. It seems perfectly possible for perfect reasoners to still be uncertain of the identity of the others, or of who belong to which group.

The results reported here update Grove and Halpern's original formulation, so to speak, by connecting it to known model-theoretic results and to neighborhood semantics, as well as by extending it with the standard notions of common and distributed knowledge. This lays the ground for studying, using modern logical tools, knowledge and beliefs in a much broader class of social situations than what was possible in standard epistemic logic, without having to move to highly complex, explicit first-order extensions of epistemic logic.

The results here suggest many possible avenues for future work. We have mentioned along the way a number of open questions related to model and proof theory, both for the original $E_n$ and $S_n$ modalities as well as for the common and distributed knowledge extensions. The joint axiomatization of the latter two notions is also open, as well as the extendability of the results here to more restricted classes of frames, in particular in those cases where the agents know their own names. Another natural next step is to look at dynamic extensions of this language~\cite{kooi2007dynamic,liberman2020dynamic}, and especially conditions under which agents can learn who are the members of certain groups. One should also study the possibility of postulating certain patterns of interactions between names. In the current framework, there is no such interaction, and one cannot even express the fact that a certain group $n$ is a sub-group of another group $m$. A first question to address here would be what extensions of the language allow for the expressions of such relations while keeping the logic decidable.

\bibliography{Bibwithdois.bib}

\begin{thebibliography}{10}
\providecommand{\bibitemdeclare}[2]{}
\providecommand{\surnamestart}{}
\providecommand{\surnameend}{}
\providecommand{\urlprefix}{Available at }
\providecommand{\url}[1]{\texttt{#1}}
\providecommand{\href}[2]{\texttt{#2}}
\providecommand{\urlalt}[2]{\href{#1}{#2}}
\providecommand{\doi}[1]{doi:\urlalt{http://dx.doi.org/#1}{#1}}
\providecommand{\bibinfo}[2]{#2}

\bibitemdeclare{article}{bennett2014organization}
\bibitem{bennett2014organization}
\bibinfo{author}{W~Lance \surnamestart Bennett\surnameend},
  \bibinfo{author}{Alexandra \surnamestart Segerberg\surnameend} \&
  \bibinfo{author}{Shawn \surnamestart Walker\surnameend}
  (\bibinfo{year}{2014}): \emph{\bibinfo{title}{Organization in the crowd: peer
  production in large-scale networked protests}}.
\newblock {\sl \bibinfo{journal}{Information, Communication \& Society}}
  \bibinfo{volume}{17}(\bibinfo{number}{2}), pp. \bibinfo{pages}{232--260},
  \doi{10.1080/1369118X.2013.870379}.

\bibitemdeclare{book}{bicchieri2005grammar}
\bibitem{bicchieri2005grammar}
\bibinfo{author}{Cristina \surnamestart Bicchieri\surnameend}
  (\bibinfo{year}{2005}): \emph{\bibinfo{title}{The grammar of society: The
  nature and dynamics of social norms}}.
\newblock \bibinfo{publisher}{Cambridge University Press},
  \doi{10.1017/CBO9780511616037}.

\bibitemdeclare{article}{brown1988}
\bibitem{brown1988}
\bibinfo{author}{M.A. \surnamestart Brown\surnameend} (\bibinfo{year}{1988}):
  \emph{\bibinfo{title}{On the logic of ability}}.
\newblock {\sl \bibinfo{journal}{Journal Phil. Log.}}
  \bibinfo{volume}{17}(\bibinfo{number}{1}), pp. \bibinfo{pages}{1--26},
  \doi{10.1007/BF00249673}.

\bibitemdeclare{book}{coleman2014hacker}
\bibitem{coleman2014hacker}
\bibinfo{author}{Gabriella \surnamestart Coleman\surnameend}
  (\bibinfo{year}{2014}): \emph{\bibinfo{title}{Hacker, hoaxer, whistleblower,
  spy: The many faces of Anonymous}}.
\newblock \bibinfo{publisher}{Verso books}.

\bibitemdeclare{book}{DEL}
\bibitem{DEL}
\bibinfo{author}{Hans \surnamestart van Ditmarsch\surnameend},
  \bibinfo{author}{Wiebe \surnamestart van~der Hoek\surnameend} \&
  \bibinfo{author}{Barteld \surnamestart Kooi\surnameend}
  (\bibinfo{year}{2007}): \emph{\bibinfo{title}{Dynamic epistemic logic}}.
\newblock \bibinfo{volume}{337}, \bibinfo{publisher}{Springer Science \&
  Business Media}, \doi{10.1007/978-1-4020-5839-4}.

\bibitemdeclare{book}{dunin2011teamwork}
\bibitem{dunin2011teamwork}
\bibinfo{author}{Barbara \surnamestart Dunin-Keplicz\surnameend} \&
  \bibinfo{author}{Rineke \surnamestart Verbrugge\surnameend}
  (\bibinfo{year}{2011}): \emph{\bibinfo{title}{Teamwork in multi-agent
  systems: A formal approach}}.
\newblock \bibinfo{volume}{21}, \bibinfo{publisher}{John Wiley \& Sons},
  \doi{10.1002/9780470665237}.

\bibitemdeclare{article}{fagin1987belief}
\bibitem{fagin1987belief}
\bibinfo{author}{Ronald \surnamestart Fagin\surnameend} \&
  \bibinfo{author}{Joseph~Y \surnamestart Halpern\surnameend}
  (\bibinfo{year}{1987}): \emph{\bibinfo{title}{Belief, awareness, and limited
  reasoning}}.
\newblock {\sl \bibinfo{journal}{Artificial intelligence}}
  \bibinfo{volume}{34}(\bibinfo{number}{1}), pp. \bibinfo{pages}{39--76},
  \doi{10.1016/0004-3702(87)90003-8}.

\bibitemdeclare{book}{RAK}
\bibitem{RAK}
\bibinfo{author}{Ronald \surnamestart Fagin\surnameend},
  \bibinfo{author}{Joseph~Y. \surnamestart Halpern\surnameend},
  \bibinfo{author}{Yoram \surnamestart Moses\surnameend} \&
  \bibinfo{author}{Moshe \surnamestart Vardi\surnameend}
  (\bibinfo{year}{2004}): \emph{\bibinfo{title}{Reasoning about Knowledge}}.
\newblock \bibinfo{publisher}{MIT press}, \doi{10.7551/mitpress/5803.001.0001}.

\bibitemdeclare{article}{grove1995naming}
\bibitem{grove1995naming}
\bibinfo{author}{Adam~J \surnamestart Grove\surnameend} (\bibinfo{year}{1995}):
  \emph{\bibinfo{title}{Naming and identity in epistemic logic Part II: a
  first-order logic for naming}}.
\newblock {\sl \bibinfo{journal}{Artificial Intelligence}}
  \bibinfo{volume}{74}(\bibinfo{number}{2}), pp. \bibinfo{pages}{311--350},
  \doi{10.1016/0004-3702(95)98593-D}.

\bibitemdeclare{article}{HG93}
\bibitem{HG93}
\bibinfo{author}{Adam~J. \surnamestart Grove\surnameend} \&
  \bibinfo{author}{Joseph~Y. \surnamestart Halpern\surnameend}
  (\bibinfo{year}{1993}): \emph{\bibinfo{title}{{Naming and Identity in
  Epistemic Logics Part I: The Propositional Case}}}.
\newblock {\sl \bibinfo{journal}{Journal of Logic and Computation}}
  \bibinfo{volume}{3}(\bibinfo{number}{4}), pp. \bibinfo{pages}{345--378},
  \doi{10.1093/logcom/3.4.345}.

\bibitemdeclare{mastersthesis}{hansen2003monotonic}
\bibitem{hansen2003monotonic}
\bibinfo{author}{Helle~Hvid \surnamestart Hansen\surnameend}
  (\bibinfo{year}{2003}): \emph{\bibinfo{title}{Monotonic modal logics}}.
\newblock Master's thesis, \bibinfo{school}{ILLC UVA}.

\bibitemdeclare{inproceedings}{hansen2007bisimulation}
\bibitem{hansen2007bisimulation}
\bibinfo{author}{Helle~Hvid \surnamestart Hansen\surnameend},
  \bibinfo{author}{Clemens \surnamestart Kupke\surnameend} \&
  \bibinfo{author}{Eric \surnamestart Pacuit\surnameend}
  (\bibinfo{year}{2007}): \emph{\bibinfo{title}{Bisimulation for neighbourhood
  structures}}.
\newblock In: {\sl \bibinfo{booktitle}{International Conference on Algebra and
  Coalgebra in Computer Science}}, \bibinfo{organization}{Springer}, pp.
  \bibinfo{pages}{279--293}, \doi{10.1007/978-3-540-73859-6_19}.

\bibitemdeclare{book}{hintikka1962knowledge}
\bibitem{hintikka1962knowledge}
\bibinfo{author}{Jaakko \surnamestart Hintikka\surnameend}
  (\bibinfo{year}{1962}): \emph{\bibinfo{title}{Knowledge and belief: An
  introduction to the logic of the two notions}}.
\newblock \bibinfo{publisher}{Cornell University Press}.

\bibitemdeclare{inproceedings}{kooi2007dynamic}
\bibitem{kooi2007dynamic}
\bibinfo{author}{Barteld \surnamestart Kooi\surnameend} (\bibinfo{year}{2007}):
  \emph{\bibinfo{title}{Dynamic term-modal logic}}.
\newblock In: {\sl \bibinfo{booktitle}{A Meeting of the Minds}}, pp.
  \bibinfo{pages}{173--186}.

\bibitemdeclare{inproceedings}{lellmann2019}
\bibitem{lellmann2019}
\bibinfo{author}{Bj{\"o}rn \surnamestart Lellmann\surnameend}
  (\bibinfo{year}{2019}): \emph{\bibinfo{title}{Combining Monotone and Normal
  Modal Logic in Nested Sequents -- with Countermodels}}.
\newblock In \bibinfo{editor}{Serenella \surnamestart Cerrito\surnameend} \&
  \bibinfo{editor}{Andrei \surnamestart Popescu\surnameend}, editors: {\sl
  \bibinfo{booktitle}{Automated Reasoning with Analytic Tableaux and Related
  Methods}}, \bibinfo{publisher}{Springer}, \bibinfo{address}{Cham}, pp.
  \bibinfo{pages}{203--220}, \doi{10.1007/978-3-030-29026-9_12}.

\bibitemdeclare{book}{lewis2008convention}
\bibitem{lewis2008convention}
\bibinfo{author}{David \surnamestart Lewis\surnameend} (\bibinfo{year}{2008}):
  \emph{\bibinfo{title}{Convention: A philosophical study}}.
\newblock \bibinfo{publisher}{John Wiley \& Sons}.

\bibitemdeclare{article}{liberman2020dynamic}
\bibitem{liberman2020dynamic}
\bibinfo{author}{Andr{\'e}s~Occhipinti \surnamestart Liberman\surnameend},
  \bibinfo{author}{Andreas \surnamestart Achen\surnameend} \&
  \bibinfo{author}{Rasmus~Kr{\ae}mmer \surnamestart Rendsvig\surnameend}
  (\bibinfo{year}{2020}): \emph{\bibinfo{title}{Dynamic term-modal logics for
  first-order epistemic planning}}.
\newblock {\sl \bibinfo{journal}{Artificial Intelligence}}
  \bibinfo{volume}{286}, p. \bibinfo{pages}{103305},
  \doi{10.1016/j.artint.2020.103305}.

\bibitemdeclare{article}{marcus1961modalities}
\bibitem{marcus1961modalities}
\bibinfo{author}{Ruth~Barcan \surnamestart Marcus\surnameend}
  (\bibinfo{year}{1961}): \emph{\bibinfo{title}{Modalities and intensional
  languages}}.
\newblock {\sl \bibinfo{journal}{Synthese}}, pp. \bibinfo{pages}{303--322},
  \doi{10.1007/BF00486629}.

\bibitemdeclare{article}{moses1988programming}
\bibitem{moses1988programming}
\bibinfo{author}{Yoram \surnamestart Moses\surnameend} \&
  \bibinfo{author}{Mark~R \surnamestart Tuttle\surnameend}
  (\bibinfo{year}{1988}): \emph{\bibinfo{title}{Programming simultaneous
  actions using common knowledge}}.
\newblock {\sl \bibinfo{journal}{Algorithmica}}
  \bibinfo{volume}{3}(\bibinfo{number}{1}), pp. \bibinfo{pages}{121--169},
  \doi{10.1007/BF01762112}.

\bibitemdeclare{article}{naumov2018everyone}
\bibitem{naumov2018everyone}
\bibinfo{author}{Pavel \surnamestart Naumov\surnameend} \& \bibinfo{author}{Jia
  \surnamestart Tao\surnameend} (\bibinfo{year}{2018}):
  \emph{\bibinfo{title}{Everyone knows that someone knows: quantifiers over
  epistemic agents}}.
\newblock {\sl \bibinfo{journal}{The review of symbolic logic}},
  \doi{10.1017/S1755020318000497}.

\bibitemdeclare{book}{pacuit2017neighborhood}
\bibitem{pacuit2017neighborhood}
\bibinfo{author}{Eric \surnamestart Pacuit\surnameend} (\bibinfo{year}{2017}):
  \emph{\bibinfo{title}{Neighborhood semantics for modal logic}}.
\newblock \bibinfo{publisher}{Springer}, \doi{10.1007/978-3-319-67149-9}.

\bibitemdeclare{inproceedings}{padmanabha2019propositional}
\bibitem{padmanabha2019propositional}
\bibinfo{author}{Anantha \surnamestart Padmanabha\surnameend} \&
  \bibinfo{author}{R~\surnamestart Ramanujam\surnameend}
  (\bibinfo{year}{2019}): \emph{\bibinfo{title}{Propositional modal logic with
  implicit modal quantification}}.
\newblock In: {\sl \bibinfo{booktitle}{Indian Conference on Logic and Its
  Applications}}, \bibinfo{organization}{Springer}, pp. \bibinfo{pages}{6--17},
  \doi{10.1007/978-3-662-58771-3_2}.

\bibitemdeclare{article}{poggiolesi2008cut}
\bibitem{poggiolesi2008cut}
\bibinfo{author}{Francesca \surnamestart Poggiolesi\surnameend}
  (\bibinfo{year}{2008}): \emph{\bibinfo{title}{A cut-free simple sequent
  calculus for modal logic S5}}.
\newblock {\sl \bibinfo{journal}{The Review of Symbolic Logic}}
  \bibinfo{volume}{1}(\bibinfo{number}{1}), pp. \bibinfo{pages}{3--15},
  \doi{10.1017/S1755020308080040}.

\bibitemdeclare{article}{roy2019shared}
\bibitem{roy2019shared}
\bibinfo{author}{Olivier \surnamestart Roy\surnameend} \& \bibinfo{author}{Anne
  \surnamestart Schwenkenbecher\surnameend} (\bibinfo{year}{2019}):
  \emph{\bibinfo{title}{Shared intentions, loose groups, and pooled
  knowledge}}.
\newblock {\sl \bibinfo{journal}{Synthese}}, pp. \bibinfo{pages}{1--19},
  \doi{10.1007/s11229-019-02355-x}.

\bibitemdeclare{incollection}{sep-collective-intentionality}
\bibitem{sep-collective-intentionality}
\bibinfo{author}{David~P. \surnamestart Schweikard\surnameend} \&
  \bibinfo{author}{Hans~Bernhard \surnamestart Schmid\surnameend}
  (\bibinfo{year}{2020}): \emph{\bibinfo{title}{{Collective Intentionality}}}.
\newblock In \bibinfo{editor}{Edward~N. \surnamestart Zalta\surnameend},
  editor: {\sl \bibinfo{booktitle}{The {Stanford} Encyclopedia of Philosophy}},
  \bibinfo{edition}{winter 2020} edition, \bibinfo{publisher}{Metaphysics
  Research Lab, Stanford University}.

\bibitemdeclare{article}{shtakser2018propositional}
\bibitem{shtakser2018propositional}
\bibinfo{author}{Gennady \surnamestart Shtakser\surnameend}
  (\bibinfo{year}{2018}): \emph{\bibinfo{title}{Propositional epistemic logics
  with quantification over agents of knowledge}}.
\newblock {\sl \bibinfo{journal}{Studia Logica}}
  \bibinfo{volume}{106}(\bibinfo{number}{2}), pp. \bibinfo{pages}{311--344},
  \doi{10.1007/s11225-017-9741-0}.

\bibitemdeclare{article}{stalnaker2006logics}
\bibitem{stalnaker2006logics}
\bibinfo{author}{Robert \surnamestart Stalnaker\surnameend}
  (\bibinfo{year}{2006}): \emph{\bibinfo{title}{On logics of knowledge and
  belief}}.
\newblock {\sl \bibinfo{journal}{Philosophical studies}}
  \bibinfo{volume}{128}(\bibinfo{number}{1}), pp. \bibinfo{pages}{169--199},
  \doi{10.1007/s11098-005-4062-y}.

\bibitemdeclare{book}{von1954essay}
\bibitem{von1954essay}
\bibinfo{author}{Georg~H \surnamestart Von~Wright\surnameend}
  (\bibinfo{year}{1954}): \emph{\bibinfo{title}{An essay in modal logic}}.
\newblock \bibinfo{publisher}{North Holland}.

\bibitemdeclare{inproceedings}{wang2018names}
\bibitem{wang2018names}
\bibinfo{author}{Yanjing \surnamestart Wang\surnameend} \&
  \bibinfo{author}{Jeremy \surnamestart Seligman\surnameend}
  (\bibinfo{year}{2018}): \emph{\bibinfo{title}{When names are not commonly
  known: epistemic logic with assignments}}.
\newblock In \bibinfo{editor}{Guram \surnamestart Bezhanishvili\surnameend},
  \bibinfo{editor}{Giovanna \surnamestart D'Agostino\surnameend},
  \bibinfo{editor}{George \surnamestart Metcalfe\surnameend} \&
  \bibinfo{editor}{Thomas \surnamestart Studer\surnameend}, editors: {\sl
  \bibinfo{booktitle}{Proceedings of AiML 2018}}.

\bibitemdeclare{inproceedings}{wang2020simpler}
\bibitem{wang2020simpler}
\bibinfo{author}{Y{\`\i}~N \surnamestart W{\'a}ng\surnameend} \&
  \bibinfo{author}{Thomas \surnamestart {\AA}gotnes\surnameend}
  (\bibinfo{year}{2020}): \emph{\bibinfo{title}{Simpler completeness proofs for
  modal logics with intersection}}.
\newblock In: {\sl \bibinfo{booktitle}{International Workshop on Dynamic
  Logic}}, \bibinfo{organization}{Springer}, pp. \bibinfo{pages}{259--276},
  \doi{10.1007/978-3-030-65840-3_16}.

\bibitemdeclare{book}{williamson2002knowledge}
\bibitem{williamson2002knowledge}
\bibinfo{author}{Timothy \surnamestart Williamson\surnameend}
  (\bibinfo{year}{2002}): \emph{\bibinfo{title}{Knowledge and its Limits}}.
\newblock \bibinfo{publisher}{Oxford University Press},
  \doi{10.1093/019925656X.001.0001}.

\end{thebibliography}
\bibliographystyle{eptcs}
\end{document}